\documentclass[12pt, draft, onecolumn]{IEEEtran}
\newcommand{\draftversion}

\usepackage{cite}
\usepackage{amsthm, amsmath,amssymb,amsfonts,bbm}
\usepackage{algorithmic}
\usepackage[dvips]{graphicx}
\usepackage{subfigure}
\graphicspath{{graphics/}}
\usepackage{multirow}
\usepackage{filecontents, pgfplots}
\usepackage[algoruled]{algorithm2e}
\usepackage{microtype}
\usepackage{tikz}
\usepackage{textcomp}
\usepackage{xcolor}
\usepackage{verbatim}
\usepackage[binary-units=true]{siunitx}
\DeclareSIUnit{\belmilliwatt}{Bm}
\DeclareSIUnit{\dBm}{\deci\belmilliwatt}
\def\BibTeX{{\rm B\kern-.05em{\sc i\kern-.025em b}\kern-.08em
		T\kern-.1667em\lower.7ex\hbox{E}\kern-.125emX}}

\makeatletter
\newif\iftag@here

\newcommand*{\taghere}[1][0pt]
{\ifmeasuring@\else
	\global\tag@heretrue
	\tikz[remember picture,overlay]{\coordinate (taghere) at (0pt,#1);}%
	\fi}

\def\place@tag{%
	\iftagsleft@
	\kern-\tagshift@
	\iftag@here
	\global\tag@herefalse
	\tikz[remember picture,overlay]%
	{\path (taghere) -| node[anchor=base]{\rlap{\boxz@}} (0pt,0pt);}%
	\else
	\if1\shift@tag\row@\relax
	\rlap{\vbox{%
			\normalbaselines
			\boxz@
			\vbox to\lineht@{}%
			\raise@tag
	}}%
	\else
	\rlap{\boxz@}%
	\fi
	\kern\displaywidth@
	\fi
	\else
	\kern-\tagshift@
	\iftag@here
	\global\tag@herefalse
	\tikz[remember picture,overlay]%
	{\path  (taghere) -|  node[anchor=base]{\llap{\boxz@}} (0pt,0pt);}%
	\else
	\if1\shift@tag\row@\relax
	\llap{\vtop{%
			\raise@tag
			\normalbaselines
			\setbox\@ne\null
			\dp\@ne\lineht@
			\box\@ne
			\boxz@
	}}%
	\else \llap{\boxz@}%
	\fi
	\fi
	\fi
}
\makeatother

\DeclareMathOperator*{\argmin}{arg\,min}
\DeclareMathOperator*{\minimize}{minimize}
\DeclareMathOperator*{\subjectto}{subject\,to}

\usepackage[acronym,nomain]{glossaries}
\makeglossaries
\newacronym{swipt}{SWIPT}{simultaneous wireless information and power transfer}
\newacronym{wpcn}{WPCN}{wireless powered communication network}
\glsunset{wpcn}
\newacronym{bs}{BS}{base station}
\newacronym{wpt}{WPT}{wireless power transfer}
\newacronym{wit}{WIT}{wireless information transfer}
\newacronym{iot}{IoT}{Internet-of-Things}
\newacronym{awgn}{AWGN}{additive white Gaussian noise}
\newacronym{tx}{TX}{transmitter}
\newacronym{ir}{IR}{information receiver}
\newacronym{eh}{EH}{energy harvesting}
\newacronym{ap}{AP}{average power}
\newacronym{pp}{PP}{peak power}

\newacronym{zf}{ZF}{zero forcing}
\newacronym{mrc}{MRC}{maximum ratio combining}
\newacronym{sinr}{SINR}{signal-to-interference-plus-noise ratio}
\newacronym{snr}{SNR}{signal-to-noise ratio}
\newacronym{tdd}{TDD}{time-division duplex}

\newacronym{siso}{SISO}{single-input single-output}
\newacronym{mimo}{MIMO}{multiple-input multiple-output}
\newacronym{miso}{MISO}{multiple-input single-output}
\newacronym{simo}{SIMO}{single-input multiple-output}

\newacronym{rf}{RF}{radio frequency}
\newacronym{dc}{DC}{direct current}
\newacronym{ac}{AC}{alternative current}
\newacronym{papr}{PAPR}{peak-to-average power ratio}
\newacronym{lp}{LPF}{low-pass filter}
\newacronym{mc}{MC}{matching circuit}
\newacronym{mrt}{MRT}{maximum ratio transmission}

\newacronym{rv}{RV}{random variable}
\newacronym{iid}{i.i.d.}{independent and identically distributed}
\newacronym{pdf}{pdf}{probability density function}

\newacronym{dnn}{DNN}{dense neural networks}
\glsunset{dnn}
\newacronym{mdp}{MDP}{Markov decision process}

\newacronym{sca}{SCA}{successive convex approximation}
\newacronym{sdr}{SDR}{semi-definite relaxation}

\newacronym{spr}{LP}{low power}
\newacronym{mpr}{MP}{medium power}
\newacronym{lpr}{HP}{high power}

\DeclareMathOperator{\rank}{rank}

\newcommand{\Tr}[1]{\text{Tr}\{#1\} }

\begin{document}
	\title{Optimal Resource Allocation and Beamforming for Two-User MISO WPCNs for a Non-linear Circuit-Based EH Model}
	
	\author{\IEEEauthorblockN{Nikita Shanin\IEEEauthorrefmark{1}, Moritz Garkisch\IEEEauthorrefmark{1}, Amelie Hagelauer\IEEEauthorrefmark{2}\IEEEauthorrefmark{3}, Robert Schober\IEEEauthorrefmark{1}, and Laura Cottatellucci\IEEEauthorrefmark{1}} \\
		\IEEEauthorblockA{\IEEEauthorrefmark{1}\textit{Friedrich-Alexander-Universit\"{a}t Erlangen-N\"{u}rnberg (FAU), Germany} \\ { \IEEEauthorrefmark{2}\textit{Fraunhofer EMFT, Germany} \\ \IEEEauthorrefmark{3}\textit{Technische Universität München, Germany} }	}}
	
\maketitle

\newtheorem{proposition}{Proposition}	
\newtheorem{lemma}{Lemma}	
\newtheorem{theorem}{Theorem}	
\newtheorem{corollary}{Corollary}
\newtheorem{assumption}{Assumption}	
\newtheorem{remark}{Remark}	
	
\begin{abstract}
	We study two-user multiple-input single-output (MISO) wireless powered communication networks (WPCNs), where the user devices are equipped with non-linear energy harvesting (EH) circuits. 
We consider time-division duplex (TDD) transmission, where the users harvest power from the signal received in the downlink phase, and then, utilize this harvested power for information transmission in the uplink phase. 
In contrast to existing works, we adopt a non-linear model of the harvested power based on a precise analysis of the employed EH circuit. 
We jointly optimize the beamforming vectors in the downlink and the time allocated for downlink and uplink transmission to minimize the average transmit power in the downlink under per-user data rate constraints in the uplink.
We provide conditions for the feasibility of the resource allocation problem and the existence of a trivial solution, respectively.
For the case where the resource allocation has a non-trivial solution, we show that it is optimal to employ no more than three beamforming vectors for power transfer in the downlink.
To determine these beamforming vectors, we develop an iterative algorithm based on semi-definite relaxation (SDR) and successive convex approximation (SCA).
Our simulation results reveal that the proposed resource allocation scheme outperforms two baseline schemes based on linear and sigmoidal EH models, respectively.
\let\thefootnote\relax\footnotetext{The work was supported in part by German Research Foundation under Project SFB 1483. }
\setcounter{footnote}{0}
\end{abstract}
\setcounter{footnote}{0}
 
	\section{Introduction}
	\label{Section:Introduction}	
	The growth of the number of low-power Internet-of-Things (IoT) devices has fuelled a significant interest in wireless powered communication networks (WPCNs) that enable energy-sustainable communication with such devices \cite{Bi2016, Clerckx2019, Liu2014, Boshkovska2017a, Zawawi2019}.
A typical \gls*{wpcn} comprises a \gls*{bs} that broadcasts a \gls*{rf} signal to the user devices in the downlink \cite{Bi2016, Clerckx2019}. 
The users are equipped with \gls*{eh} circuits that extract the received \gls*{rf} energy. The collected energy is then used for information transmission in the uplink \cite{Bi2016, Clerckx2019}.

A WPCN employing multiple antennas at the \gls*{bs} was considered in \cite{Liu2014}. 
The authors optimized the beamforming vectors and transmit powers for maximization of the total throughput in the uplink.
Although the results in \cite{Liu2014} provide insights for WPCN design, they are based on a linear EH model, whereas the experiments in \cite{Kim2020} and \cite{Sakai2021} showed that typical EH circuits are non-linear.
In order to take the non-linearities of EH circuits into account, the authors in \cite{Boshkovska2015} proposed a practical \gls*{eh} model based on a parametrized sigmoidal function of the average power of the received RF signal.
This model was utilized for the analysis of multi-antenna WPCNs in \cite{Boshkovska2017a}, where the authors optimized the covariance matrix of the transmit symbol vectors in the downlink for maximization of the sum information rate in the uplink.
Since the EH model in \cite{Boshkovska2015} characterizes only the \emph{average} harvested power at the user devices, the WPCN design in \cite{Boshkovska2017a} does not fully capture the non-linearity of the EH circuits.
Therefore, based on an accurate analysis of a typical EH circuit, the authors in \cite{Morsi2019} developed a precise EH model that maps the \emph{instantaneous} received RF power to the \emph{instantaneous} harvested power.
Based on this model, the authors in \cite{Shanin2021c} considered multi-user multi-antenna \gls*{wpt} systems and showed that the optimal transmit strategy that maximizes a weighted sum of the average harvested powers at EH nodes employs multiple beamforming vectors at the BS.

The main contributions of this paper can be summarized as follows.
We consider a two-user \gls*{miso} WPCN, where the EH at the user devices is characterized by the non-linear circuit-based EH model developed in \cite{Morsi2019}.
We formulate an optimization problem to minimize the average transmit power in the downlink under per-user rate constraints in the uplink.
We provide conditions for the feasibility of the optimization problem and the existence of a trivial solution, respectively.
Next, when the problem is feasible and has a non-trivial solution, we show that three beamforming vectors are sufficient for optimal transmission in the downlink.
To determine these beamforming vectors, we design an iterative algorithm based on \gls*{sdr} and \gls*{sca} \cite{Sun2017}.
Our simulation results demonstrate that the proposed approach significantly outperforms two baseline schemes that are based on the linear and sigmoidal EH models in \cite{Liu2014} and \cite{Boshkovska2017a}, respectively.
\begin{figure}[!t]
	\includegraphics[width=\linewidth, draft=false]{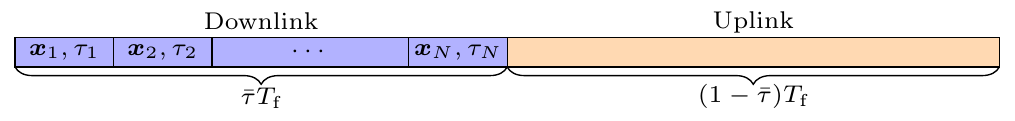}
	\caption{Structure of a time frame of length $T_\text{f}$.}
	\label{Fig:IllustrationTimeSharing}
\end{figure}

Throughout this paper, we use the following notations.
Bold upper case letters $\boldsymbol{X}$ represent matrices.
Bold lower case letters $\boldsymbol{x}$ stand for vectors and ${x}_{i}$ is the $i^\text{th}$ element of $\boldsymbol{x}$.
$\boldsymbol{X}^H$ denotes the Hermitian transpose of matrix $\boldsymbol{X}$.
$\| {\boldsymbol{x}} \|_2$ represents the L$_2$-norm of $\boldsymbol{x}$.
The sets of complex and real numbers are denoted by $\mathbb{C}$ and $\mathbb{R}$, respectively.
The real part of a complex number is denoted by $\mathcal{R}\{\cdot\}$.
$\boldsymbol{I}_K$ stands for the square identity matrix of size $K$.
The imaginary unit is denoted by $j$.
	
	\section{System Model}
	\label{Section:SystemModel}
		We consider a \gls*{miso} \gls*{wpcn}, where $K = 2$ single-antenna users are equipped with \gls*{eh} circuits \cite{Morsi2019} and $N_\text{t}\geq K$ antennas are employed at the BS.
We adopt \gls*{tdd} transmission and assume that each time frame of length $T_\text{f}$ is divided into two subframes, as shown in Fig.~\ref{Fig:IllustrationTimeSharing}. 
In the first subframe of length $\bar{\tau} T_\text{f}$, $\bar{\tau} \in [0,1)$, the BS transmits a power-carrying \gls*{rf} signal to the user devices, which harvest the received power.
In the subsequent subframe of length $(1-\bar{\tau}) T_\text{f}$, this harvested power is utilized for information transmission in the uplink.
We assume that the channels between the BS and the user devices are constant for the duration of a time frame.
Furthermore, we denote the channel between the BS and user $k\in\{1,2\}$ by $\boldsymbol{h}_k \in \mathbb{C}^{N_\text{t}}$ and assume that it is perfectly known at the BS.

		\subsection{Downlink Phase}
		In the downlink, the BS broadcasts a pulse-modulated RF signal, whose equivalent complex baseband representation is modelled as $\boldsymbol{x}(t) = \sum_{n=1}^{N} \boldsymbol{x}_n \psi_n(t)$, where $\psi_n(t) = \Pi\big(\frac{ t/T_\text{f} - \sum_{k=0}^{n-1}\tau_{k} }{ \tau_{n} } \big)$ is the transmit pulse, $\tau_0 = 0$,  $\Pi(t)$ is a rectangular function that takes value $1$ if $t \in [0,1)$ and $0$, otherwise, $\boldsymbol{x}_n \in \mathbb{C}^{N_\text{t}}$ is the symbol vector transmitted in time slot $n\in\{1,2, \cdots, N\}$, and $N$ is the number of employed symbol vectors, see Fig.~\ref{Fig:IllustrationTimeSharing}.
Here, $\tau_n$ is the portion of the time frame of length $T_\text{f}$ utilized for the transmission of symbol vector $\boldsymbol{x}_n$ with $\sum_{n=1}^{N} \tau_n = \bar{\tau}$.
Thus, the \gls*{rf} signal received at user $k\in\{1,2\}$ is given by $z^{\text{RF}}_k(t) = \sqrt{2} \mathcal{R} \Big\{ \boldsymbol{h}_k^H \boldsymbol{x}(t) \exp(j 2 \pi f_c t)  \Big\}$, where $f_c$ denotes the carrier frequency.
Similar to \cite{Morsi2019}, the noise received at the users is neglected since its contribution to the harvested power is negligible.

To harvest power, the users are equipped with identical non-linear \gls*{eh} circuits.
We denote the power harvested at user $k$ in time slot $n$ by $\rho_{k,n}$ and, as in \cite{Morsi2019} and \cite{Shanin2021a}, model it by a non-linear monotonic non-decreasing function $\phi(\cdot)$ of the instantaneous received power. 
Thus, $\rho_{k,n} = \phi\big(|z_{k,n}|^2\big)$ with  $z_{k,n} = \boldsymbol{h}_k^H \boldsymbol{x}_n$ and $\phi(\cdot)$ is given by \cite{Morsi2019}:
\begin{align}
	\phi\big(|z|^2\big) &= \min \big\{ \varphi(|z|^2), \varphi(A_s^2) \big\}, \label{Eqn:ModelEH} \\
	\varphi(|z|^2) &= \lambda \Big[\mu^{-1} W_0\Big( \mu \exp(\mu) I_0 \Big( \nu \sqrt{2 |z|^2} \Big) - 1 \Big)^2\Big],
\end{align}
\noindent where $W_0(\cdot)$ is the principle branch of the Lambert-W function, $I_0(\cdot)$ is the modified Bessel function of the first kind and order zero, and $\lambda, \mu$, and $\nu$ are parameters of the EH circuit that depend on the circuit elements but not on the received signal.
Since practical EH circuits are driven into saturation for large input powers \cite{Boshkovska2015, Morsi2019}, $\phi(\cdot)$ in (\ref{Eqn:ModelEH}) is bounded, i.e., $\phi(|z|^2) \leq \phi(|A_s|^2), \; \forall z \in \mathbb{C}$, where $A_s$ is the minimum input signal magnitude level at which the output power starts to saturate.
Hence, the average power harvested at user $k\in\{1,2\}$ in the downlink can be expressed as follows:
\begin{equation}
	{p}^\text{d}_k =\sum_{n=1}^{N} \tau_n \phi\big(|\boldsymbol{h}_k^H \boldsymbol{x}_n|^2\big)
	\label{Eqn:AverageHarvPower}
\end{equation}
Finally, we assume that user $k$ is equipped with a rechargeable built-in battery having initial energy $q_k$, which is known to the BS.
Thus, at the end of the downlink phase, the amount of energy available at user $k$ is given by $E_k = q_k + p^\text{d}_k T_f$.
		\label{Section:DownlinkStage}
		\subsection{Uplink Phase}
		In the uplink, user $k\in\{1,2\}$ transmits information symbols $s_k$ with zero mean and unit variance to the BS utilizing a portion of the available energy $E_k$.
Assuming uplink-downlink reciprocity of the channel, the symbol vector $\boldsymbol{r} \in \mathbb{C}^{N_\text{t}}$ received at the BS in the scheduled time slot is given by
\begin{equation}
	\boldsymbol{r} = \sum_{k=1}^{K} \boldsymbol{h}_k \sqrt{p^\text{u}_k} s_k + \boldsymbol{n},
\end{equation}
where $p^\text{u}_k$ is the power utilized by user $k$ for information transmission and $\boldsymbol{n} \in \mathbb{C}^{N_\text{t}}$ is an \gls*{awgn} vector with zero mean and covariance matrix $\sigma^2 \boldsymbol{I}_{N_\text{t}}$.
Since we have $N_\text{t} \geq K$, we can adopt \gls*{zf} equalization to suppress the inter-user interference at the BS.
Hence, the detected information symbol $\hat{s}_k$ of user $k$ can be expressed as follows:
\begin{equation}
	\hat{s}_k = \boldsymbol{f}_{\hspace*{-2pt}k} \boldsymbol{r} = \sqrt{p^\text{u}_k} s_k + \tilde{n}_k,
\end{equation}
\noindent where $\tilde{n}_k = \boldsymbol{f}_{\hspace*{-2pt}k} \boldsymbol{n}$ is the equivalent \gls*{awgn} with variance $\tilde{\sigma}^2_k = \|\boldsymbol{f}_{\hspace*{-2pt}k}\|^2_2\sigma^2$ for user $k$ at the BS.
Here, equalization vector $\boldsymbol{f}_{\hspace*{-2pt}k} \in \mathbb{C}^{N_\text{t}}$ is the $k^\text{th}$ row of matrix $\boldsymbol{F} = (\boldsymbol{H}^H \boldsymbol{H})^{-1} \boldsymbol{H}^H$ with $\boldsymbol{H} = [\boldsymbol{h}_1 \; \boldsymbol{h}_2]$.
Finally, the data rate of user $k$ is given by $ R_k = (1-\bar{\tau}) \log_2(1 + \Gamma_k)$, where $\Gamma_k = { p^\text{u}_k}/{\tilde{\sigma}^2 }$ is the \gls*{snr}.
		\label{Section:UplinkStage}
		
	\section{Problem Formulation and Solution}
	In this section, we develop a resource allocation algorithm for the considered MISO WPCN.
	\subsection{Problem Formulation}
		We formulate the following non-convex optimization problem:
\begin{subequations}
	\begin{align}
		\minimize_{\substack{\boldsymbol{\tau}, \bar{\tau} \in [0,1), \\ \boldsymbol{X}, \boldsymbol{p}^\text{u}, N}  } \; &P_{\text{DL}} \label{Eqn:ProblemObj} \\
		\subjectto \; & R_k \geq R^{\text{req}}_k, \forall k \label{Eqn:ProblemC1} \\
		& (1-\bar{\tau})p^\text{u}_k T_f \leq E_k, \forall k  \label{Eqn:ProblemC2} \\
		&\sum_{n=1}^{N} \tau_n = \bar{\tau}, \label{Eqn:ProblemC3}
	\end{align}
	\label{Eqn:OriginalProblem}
\end{subequations}
\noindent where we jointly optimize $\bar{\tau}$, $\boldsymbol{\tau} = [\tau_1, \tau_2 \cdots \tau_N]$, the uplink transmit powers $\boldsymbol{p}^u = [p^u_1, p^u_2]$, symbol vectors $\boldsymbol{x}_n$ collected in $\boldsymbol{X} = [\boldsymbol{x}_1 \; \boldsymbol{x}_2 \cdots \boldsymbol{x}_N]$, and the number of time slots $N$.
In (\ref{Eqn:OriginalProblem}), we minimize the average transmit power in the downlink $P_{\text{DL}} = \sum_n \tau_n \|\boldsymbol{x}_n\|_2^2$ under per-user rate constraints in the uplink in (\ref{Eqn:ProblemC1}), where $R_k^\text{req}$ is the minimum required rate of user $k$.
Here, (\ref{Eqn:ProblemC2}) and (\ref{Eqn:ProblemC3}) ensure that the transmit energy consumed by user $k$ in the uplink does not exceed $E_k$ and the obtained resource allocation is feasible, respectively.
We note that in contrast to the WPCN design in \cite{Liu2014} and \cite{Boshkovska2017a}, where covariance matrix $\tilde{\boldsymbol{X}} = \sum_{n=1}^{N} \frac{\tau_n}{\bar{\tau}} \boldsymbol{x}_n \boldsymbol{x}_n^H$ was optimized, the EH model in (\ref{Eqn:ModelEH}) characterizes the instantaneous harvested power and enables the optimization of individual transmit symbol vectors $\boldsymbol{x}_n$ in (\ref{Eqn:OriginalProblem}).
	\subsection{Characterization of Optimal Solution}
		In the following propositions, we provide a feasibility condition and characterize the optimal solution of (\ref{Eqn:OriginalProblem}).
\begin{proposition}
	Problem (\ref{Eqn:OriginalProblem}) is feasible if and only if $\exists \bar{\tau} \in [0,1)$, such that the following condition holds $\forall k \in\{1,2\}$\text{\upshape :}
	\begin{equation}
		f_k(\bar{\tau}) \triangleq \frac{1-\bar{\tau}}{\bar{\tau}} \Big[2^\frac{R^{\text{req}}_k}{1-\bar{\tau}} - 1\Big]{\tilde{\sigma}^2} - \frac{q_k}{T_f \bar{\tau}} \leq \phi(A_s^2). 
	\end{equation}
\end{proposition}
\begin{proof}
	Please refer to Appendix~\ref{Appendix:ProofProposition1}.
\end{proof}

\begin{proposition}
	If the following conditions hold\text{\upshape :}
	\begin{equation}
		\big[2^{R^{\text{\upshape{req}}}_k} - 1\big]{\tilde{\sigma}^2} < \frac{q_k}{T_f}, \forall k \in \{1,2\},
		\label{Eqn:PropositionCondition}
		\ifdefined\draftversion\else\vspace*{-10pt}\fi	
	\end{equation}
	\noindent problem (\ref{Eqn:OriginalProblem}) is feasible and the optimal solution is trivial\footnotemark\text{\upshape:} $\bar{\tau}^* = 0, N^* = 0$, ${\tau}^* = 0$, $\boldsymbol{x}^* = \boldsymbol{0}$, ${p}^{\text{u}^*}_k = \big[2^{R^{\text{\upshape{req}}}_k} - 1\big]{\tilde{\sigma}^2}, \forall k \in \{1,2\}$.
\end{proposition}
\begin{proof}
	Please refer to Appendix~\ref{Appendix:ProofProposition2}.
\end{proof}
\begin{proposition}
	If problem (\ref{Eqn:OriginalProblem}) is feasible and conditions (\ref{Eqn:PropositionCondition}) do not hold, the optimal number of transmit symbol vectors is $N^* \leq 3$ and $\bar{\tau}^* \in [\bar{\tau}^\text{\upshape{min}}, \bar{\tau}^\text{\upshape{max}}]$ with $\bar{\tau}^\text{\upshape{max}} = \max\{\bar{\tau}^\text{\upshape{max}}_1,  \bar{\tau}^\text{\upshape{max}}_2\}$ and $\bar{\tau}^\text{\upshape{min}} = \max\{\bar{\tau}^\text{\upshape{min}}_1, \bar{\tau}^\text{\upshape{min}}_2\}$.
	Here, for $k \in \{1,2\}$, $\bar{\tau}^\text{\upshape{max}}_k$ is the solution of the following equation\text{\upshape:}
	\begin{equation}
		2^\frac{R^{\text{\upshape{req}}}_k}{1-\bar{\tau}_k^\text{\upshape{max}}} \ln2 R^{\text{\upshape{req}}}_k  {\tilde{\sigma}^2} =   f_k(\bar{\tau}_k^\text{\upshape{max}}) + \frac{q_k}{T_f}
		\ifdefined\draftversion\else\vspace*{-10pt}\fi	
	\end{equation}
	and $\bar{\tau}^\text{\upshape{min}}_k = \min \{ \bar{\tau}: f(\bar{\tau}) = \phi(A_s^2) \}$.
\end{proposition}
\begin{proof}
	Please refer to Appendix~\ref{Appendix:ProofProposition3}.
\end{proof}

Proposition~1 shows that problem (\ref{Eqn:OriginalProblem}) is feasible if and only if the power required for uplink transmission with rate $R^\text{req}_k$ is not greater than $\phi(A_s^2)$ for all $k \in \{1,2\}$.
Next, Proposition~2 reveals that no power transfer in the downlink is needed if the available power $q_k$ is sufficient for transmission at the data rate required for each user.
Finally, as shown in Proposition~3, if the problem is feasible and the solution is non-trivial, the optimal number of time slots satisfies $N^* \leq 3$ and the optimal length of the downlink subframe $\bar{\tau}^* \in [\bar{\tau}^\text{\upshape{min}}, \bar{\tau}^\text{\upshape{max}}]$.
In the following, we consider the case, where problem (\ref{Eqn:OriginalProblem}) is feasible and has a non-trivial solution.
\footnotetext{We note that the trivial solution of problem (\ref{Eqn:OriginalProblem}) is not unique. Here, we provide the energy efficient solution with the minimum feasible ${p}^{\text{u}}_k, \forall k$.}

We note that the functions that appear in (\ref{Eqn:OriginalProblem}) are invariant with respect to a scalar phase rotation of the symbol vectors in $\boldsymbol{X}$.
Hence, for the solution of (\ref{Eqn:OriginalProblem}), the transmit symbol vectors can be decomposed as $\boldsymbol{x}_n = \boldsymbol{w}_{n} {d}_n$, $n \in \{1,2,3\}$, where $\boldsymbol{w}_{n}$ are beamforming vectors rotated by unit-norm symbols ${d}_n = \exp(j \theta_n)$ with arbitrary phases $\theta_n$.
	\subsection{Reformulation of Optimization Problem}
		Next, we determine the optimal fraction $\bar{\tau}^*$ using a one-dimensional grid search as $\bar{\tau}^* = \argmin_{\bar{\tau} \in[\bar{\tau}^\text{\upshape{min}}, \bar{\tau}^\text{\upshape{max}}]} P_\text{DL}^*(\bar{\tau})$, where
\begin{equation}
	P_\text{DL}^*(\bar{\tau}) = \min_{\mathcal{F}} \{P_{\text{DL}} \}
	\label{Eqn:ProblemRef1}
\end{equation}
with $\mathcal{F} = \{\bar{\boldsymbol{W}}, \boldsymbol{p}^\text{u}, \boldsymbol{\tau}:\text{(\ref{Eqn:ProblemC1})}, \text{(\ref{Eqn:ProblemC2})}, \text{(\ref{Eqn:ProblemC3})} \}$ and $\bar{\boldsymbol{W}}= [\boldsymbol{w}_1 \; \boldsymbol{w}_2 \; \boldsymbol{w}_3 ]$.
In order to solve (\ref{Eqn:ProblemRef1}), we equivalently reformulate constraints (\ref{Eqn:ProblemC1}) and (\ref{Eqn:ProblemC2}) as follows:
\begin{align}
	p^\text{u}_k &\geq \rho_k^{\text{req}}, \; \forall k, \label{Eqn:ProblemC1Ref} \\
	p^\text{u}_k &\leq \bar{q}_k/T_\text{f} + p^\text{d}_k \bar{\bar{\tau}}, \; \forall k, \label{Eqn:ProblemC2Ref}
\end{align}
\noindent respectively, where $\rho_k^{\text{req}} = (2^{\gamma^{\text{req}}_k} - 1) \tilde{\sigma}^2_k$, $\bar{q}_k = \frac{q_k}{(1-\bar{\tau})}$, and $\bar{\bar{\tau}} = (1-\bar{\tau})^{-1}$.
Here, $\gamma^{\text{req}}_k = {\frac{R_k^\text{req}}{1-\bar{\tau}}}$ is the equivalent rate required by user $k$.
Thus, we can equivalently reformulate optimization problem (\ref{Eqn:ProblemRef1}) as follows:
\begin{subequations}
	\begin{align}
		\minimize_{\substack{\boldsymbol{\beta}, \boldsymbol{p}^\text{u}, \boldsymbol{W}_{\hspace*{-2pt}1}, \\ \boldsymbol{W}_{\hspace*{-2pt}2}, \boldsymbol{W}_{\hspace*{-2pt}3}} } & \; \bar{{\tau}} \sum_{n=1}^{3} \beta_n \Tr{\boldsymbol{W}_{\hspace*{-2pt}n}}
		\label{Eqn:ProblemRef1Obj} \\
		\subjectto \; &\;\;	\rank\{\boldsymbol{W}_{\hspace*{-2pt}n}\}\leq1, \boldsymbol{W}_{\hspace*{-2pt}n} \in \mathcal{S}_{+}, \forall n \label{Eqn:ProblemRef1C1} \\
		&\;\; \text{(\ref{Eqn:ProblemC3})}, \text{(\ref{Eqn:ProblemC1Ref})}, \text{(\ref{Eqn:ProblemC2Ref})} \nonumber
	\end{align}
	\label{Eqn:ProblemRef2}
\end{subequations}
\noindent\hspace*{-3pt}where  $\boldsymbol{\beta} = [\beta_1, \beta_2, \beta_3]^\top$, $\beta_n = \frac{\tau_n}{\bar{\tau}} \geq 0$, $\boldsymbol{W}_{\hspace*{-2pt}n} = \boldsymbol{w}_{\hspace*{-2pt}n} \boldsymbol{w}_{\hspace*{-2pt}n}^H, n\in\{1,2,3\}$, and $\mathcal{S}_{+} \subset \mathbb{C}^{N_\text{t} \times N_\text{t}}$ denotes the set of positive semidefinite matrices.
Here, the harvested power is given by $p^\text{d}_k = \bar{\tau} \sum_{n=1}^{3} \beta_n \phi(\boldsymbol{h}_k^H \boldsymbol{W}_{\hspace*{-2pt}n} \boldsymbol{h}_k)$.
Problem (\ref{Eqn:ProblemRef2}) is still non-convex due to the non-convexity of constraints (\ref{Eqn:ProblemC2Ref}) and (\ref{Eqn:ProblemRef1C1}).
		\subsection{Suboptimal Iterative Solution}
		In the following, we design an iterative algorithm to solve (\ref{Eqn:ProblemRef2}).
First, we drop the rank-1 constraint in (\ref{Eqn:ProblemRef1C1}).
Next, we note that the EH circuits of the user devices are not driven into saturation if
\begin{equation}
	\boldsymbol{h}^H_k \boldsymbol{W}_{\hspace*{-2pt}n} \boldsymbol{h}_k \leq A_s^2, \; \forall n, k.
	\label{Eqn:SaturationConstr}
\end{equation}
Furthermore, if condition (\ref{Eqn:SaturationConstr}) holds, function $\phi(\cdot)$ is convex and, thus, problem (\ref{Eqn:ProblemRef2}) can be efficiently solved via SCA \cite{Sun2017}.
Therefore, in the following, we assume that the EH circuits are not driven into saturation\footnotemark, i.e., condition (\ref{Eqn:SaturationConstr}) holds.
\footnotetext{Alternatively, as in \cite{Shanin2021c}, one can adopt an exhaustive search exhibiting high computational complexity to identify the users that are driven into saturation.}
Thus, in iteration $t$ of the algorithm, we linearize convex function $\phi(\cdot)$ in (\ref{Eqn:ProblemC2Ref}) as follows \cite{Sun2017}:
\begin{equation}
	\begin{aligned}
		\beta_n \phi(\boldsymbol{h}^H_k \boldsymbol{W}_{\hspace*{-2pt}n} \boldsymbol{h}_k) & \geq \phi'(\boldsymbol{h}^H_k \boldsymbol{W}^{(t)}_{\hspace{-2pt}n} \boldsymbol{h}_k) \Tr{\boldsymbol{H}_k \boldsymbol{V}_{\hspace*{-2pt}n}} + \\ & \beta_n \hat{\psi}(\boldsymbol{W}^{(t)}_{\hspace{-2pt}n}) \triangleq \Psi_k(\beta_n, \boldsymbol{V}_{\hspace*{-2pt}n}; \boldsymbol{W}^{(t)}_{\hspace{-2pt}n}),
	\end{aligned}
\end{equation}
\noindent where $\hat{\psi}(\boldsymbol{W}^{(t)}_{\hspace{-2pt}n}) = \phi(\boldsymbol{h}^H_k \boldsymbol{W}^{(t)}_{\hspace{-2pt}n} \boldsymbol{h}_k) - \phi'(\boldsymbol{h}^H_k \boldsymbol{W}^{(t)}_{\hspace{-2pt}n} \boldsymbol{h}_k) \Tr{\boldsymbol{H}_k \boldsymbol{W}^{(t)}_{\hspace{-2pt}n}}$, $\boldsymbol{V}_{\hspace*{-2pt}n} = \beta_n \boldsymbol{W}_{\hspace*{-2pt}n}$, $\boldsymbol{W}^{(t)}_{\hspace{-2pt}n} = {\boldsymbol{V}^{(t)}_{\hspace*{-2pt}n}}/{\beta^{(t)}_n}$, $\boldsymbol{H}_k = \boldsymbol{h}_k \boldsymbol{h}^H_k$, and $\beta^{(t)}_n$ and $\boldsymbol{W}^{(t)}_{\hspace{-2pt}n}$ are the values of $\beta_n$ and matrix $\boldsymbol{W}_{\hspace*{-2pt}n}$ obtained in iteration $t-1$ of the algorithm, respectively.
Thus, we reformulate constraint (\ref{Eqn:ProblemC2Ref}) as follows:
\begin{align}
	p^\text{u}_k - \bar{\bar{\tau}} \bar{\tau} \sum_{n=1}^{3} \Psi_k(\beta_n, \boldsymbol{V}_{\hspace*{-2pt}n}; \boldsymbol{W}^{(t)}_{\hspace{-2pt}n}) \leq \bar{q}_k/T_\text{f}
	\label{ProblemC2Ref2}
\end{align}
Finally, in iteration $t$ of the algorithm, we solve the following optimization problem:
\begin{subequations}
	\begin{align}
		\minimize_{\substack{\boldsymbol{V}_{\hspace*{-2pt}1}, \boldsymbol{V}_{\hspace*{-2pt}2}, \boldsymbol{V}_{\hspace*{-2pt}3} \\ \boldsymbol{\beta}, \boldsymbol{p}^\text{u}}}  \quad&\; \bar{{\tau}} \sum_{n=1}^{3} \Tr{\boldsymbol{V}_{\hspace*{-2pt}n}} 
		\label{Eqn:AlgorithmProblemObj} \\
		\subjectto \quad&\; \boldsymbol{V}_{\hspace*{-2pt}n} \in \mathcal{S}_{+}, \; \forall n
		\label{Eqn:AlgorithmProblemC3} \\
		&\;\text{(\ref{Eqn:ProblemC3})}, \text{(\ref{Eqn:ProblemC1Ref})}, \text{(\ref{Eqn:SaturationConstr})}, \text{(\ref{ProblemC2Ref2})} \nonumber
	\end{align}
	\label{Eqn:AlgorithmProblem}
\end{subequations}
\noindent Optimization problem (\ref{Eqn:AlgorithmProblem}) is convex and, hence, can be solved via a convex solver, such as CVX \cite{Grant2015}.
In the following proposition, we show that although we dropped the rank-1 constraint in (\ref{Eqn:ProblemRef1C1}) , the solution of (\ref{Eqn:AlgorithmProblem}) satisfies $\rank\{\boldsymbol{V}_{\hspace*{-2pt}n}^*\}\leq~1$.
\begin{proposition}
	The solution of (\ref{Eqn:AlgorithmProblem}) satisfies $\rank\{\boldsymbol{V}_{\hspace*{-2pt}n}^*\}\leq~1, \forall n \in \{1,2,3\}$.
\end{proposition}
\begin{proof}
	Please refer to Appendix~\ref{Appendix:ProofProposition4}.
\end{proof}
Hence, as solution of (\ref{Eqn:ProblemRef2}), we obtain $\beta_n^* = \beta_n^{(t)}$ and $\boldsymbol{w}_{n}^* = \lambda_n \hat{\boldsymbol{w}}_{n}^* $, where $\lambda_n$ and $\hat{\boldsymbol{w}}_{n}^*$ are the non-zero eigenvalue and the corresponding eigenvector of $\boldsymbol{W}_{\hspace*{-2pt}n}^* = {\boldsymbol{V}^{(t)}_{\hspace*{-2pt}n}}/{\beta^{(t)}_n}, \, \forall n$, respectively.
The proposed algorithm\footnotemark\hspace*{0pt} is summarized in Algorithm~\ref{Algorithm}.
Since the channel gains $\boldsymbol{h}_k$, the required rates $R_k^\text{req}$, and the initial energies $q_k, \forall k,$ are known to the BS, problem (\ref{Eqn:OriginalProblem}) can be solved at the BS with Algorithm~1.
\footnotetext{We note that the extension of Algorithm~1 to the multi-user case is possible.}

\begin{algorithm}[!t]	
	\small				
	\ifdefined\draftversion
	\linespread{1.5}\selectfont
	\else
	\linespread{1.0}\selectfont
	\fi	
	\SetAlgoNoLine%
	\SetKwFor{Foreach}{for each}{do}{end}		
	Initialize: Required rates $R_k^\text{req}$, initial energies $q_k, \forall k$, and error tolerances $\epsilon_\text{SCA}$, $\epsilon_{\tau}$.	\\	
	1. Set initial values $\bar{\tau} = \bar{\tau}^\text{\upshape{min}}$, $i = 1$.\\
	\Repeat{$\tau \geq \bar{\tau}^\text{\upshape{max}}$}{		
		2. Determine required harvested powers $\rho_k^{\text{req}}$ \\
		3. Randomly initialize $\boldsymbol{V}_{\hspace*{-2pt}n}^{(0)}, \forall n, \boldsymbol{\beta}^{(0)}$, set $t = 0$.\\
		\Repeat{$|h^{(t)}-h^{(t-1)}|\leq \epsilon_\text{\upshape SCA}$ }{		
			a. For given $\boldsymbol{V}_{\hspace*{-2pt}n}^{(t)}, \boldsymbol{\beta}^{(t)}$, obtain $\boldsymbol{p}^\text{u}, \boldsymbol{V}_{\hspace*{-2pt}n}^{(t+1)}, \boldsymbol{\beta}^{(t+1)}$ as the solution of (\ref{Eqn:AlgorithmProblem}) \\								
			b. Evaluate $h^{(t+1)} = \bar{{\tau}} \sum_{n=1}^{3} \Tr{\boldsymbol{V}_{\hspace*{-2pt}n}}$\\
			c. Determine $\boldsymbol{W}^{(t)} = {\boldsymbol{V}_{\hspace*{-2pt}n}^{(t)}}/{\beta_n^{(t)}}$\\
			d. Set $t = t+1$\\ 
		}
		4. Obtain $\boldsymbol{w}_{n}^* = \lambda_n \hat{\boldsymbol{w}}_{n}^*, \forall n$\\
		5. Store $\bar{\boldsymbol{W}}_{\hspace*{-2pt}i} = [\boldsymbol{w}_1^*  \boldsymbol{w}_2^* \boldsymbol{w}_{3}^*]$, $\boldsymbol{\beta}_i = \boldsymbol{\beta}^{(t-1)}$, $\boldsymbol{p}^\text{u}_i = \boldsymbol{p}^\text{u}$, $\xi_i= \bar{\tau} \sum_n \beta_n^{(t-1)} \| \boldsymbol{w}^*_n\|_2^2$\\
		6. Set $\bar{\tau} = \bar{\tau} + \epsilon_{\tau}$, $i = i+1$.
	}
	7. Find index $i^*$ yielding the minimum $\boldsymbol{\xi}$ \\
	\textbf{Output:} 
	$ \bar{\tau}^* = \bar{\tau}^\text{\upshape{min}} + (i^*-1) \epsilon_{\tau}$, $\boldsymbol{\beta}_{i^*}$, $\bar{\boldsymbol{W}}_{\hspace*{-2pt}i^*}$, $\boldsymbol{p}^\text{u}_{i^*}$
	\caption{\strut Algorithm for solving (\ref{Eqn:OriginalProblem}) }
	\label{Algorithm}
\end{algorithm}	
\setlength{\textfloatsep}{3pt}

	\section{Numerical Results}
	In this section, we evaluate the performance of the proposed resource allocation scheme via simulations.
In order to enable reliable communication, we assume that the BS and each user device have a line-of-sight link.
The noise variance is set to $\sigma^2 = \SI{-110}{\deci\belmilliwatt}$.
We compute the path losses as $\big(\frac{c_l}{4 \pi d_k f_c}\big)^2$, where $f_c =\SI{868}{\mega\hertz}$, and $c_l$ and $d_k = \SI{10}{\meter}, k \in \{1,2\},$ are the speed of light and the distance between the BS and user device $k$, respectively.
Furthermore, the channel gains $\boldsymbol{h}_k$ follow Ricean distributions with Ricean factor $1$.
For the EH model in (\ref{Eqn:ModelEH}), we adopt parameter values $\mu = 1.85, \nu = 2.2\cdot10^3, \lambda = 2.5\cdot10^{-7}, A_s^2 = 2\cdot10^{-4}$ as in \cite{Morsi2019}.
For Algorithm~1, we adopt error tolerances $\epsilon_{\text{SCA}} = 10^{-4}$ and $\epsilon_{\tau} = 0.1$ and the initial energy for each user is set to $q_k = \SI{0}{\joule}$.
All simulation results are averaged over 100 channel realizations.
\ifdefined\draftversion
\begin{figure}[!t]
	\centering
	\ifdefined\draftversion
	\includegraphics[width=0.8\linewidth, draft=false]{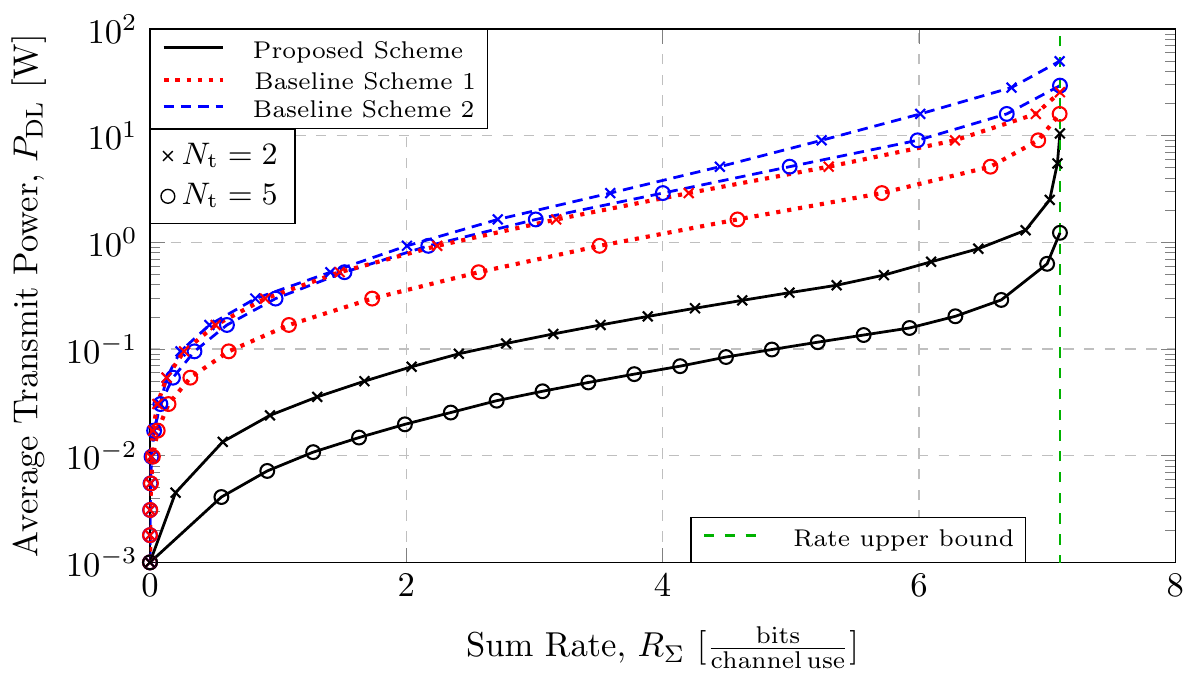}
	\else
	\includegraphics[width=0.8\linewidth, draft=false]{SimResults}
	\fi	
	\ifdefined\draftversion\else\vspace*{-10pt}\fi	
	\caption{Average transmit powers for different required rates and numbers of BS antennas.}
	\label{Fig:SimResults}
	\ifdefined\draftversion\else\vspace*{10pt}\fi	
\end{figure}
\else\fi

In Fig.~\ref{Fig:SimResults}, we show the average transmit power $P_{\text{DL}}$ as a function of sum rate $R_\Sigma$, where $R_\Sigma = \sum_k R_k^\text{req}$ and $R_1^\text{req} = R_2^\text{req}$.
As Baseline Scheme~1 and Baseline Scheme~2, we adopt the WPCN designs in \cite{Boshkovska2017a} and \cite{Liu2014}, which are based on the sigmoidal and linear EH models, respectively.
Since time fraction $\bar{{\tau}}$ and covariance matrix $\tilde{\boldsymbol{X}}$ were optimized in \cite{Liu2014} and \cite{Boshkovska2017a}, for the baseline schemes, in the downlink, we adopt $\tilde{N} = \rank \{\tilde{\boldsymbol{X}}\}$ symbol vectors with $\tau_n = {\bar{{\tau}}}/{\tilde{N}}, n\in\{1,2,\cdots, \tilde{N}\},$ and obtain these vectors from the dominant eigenvectors of $\tilde{\boldsymbol{X}}$.
First, we observe that for each considered system setup, the sum rates are bounded from above (indicated by the dashed green line) since the EH circuits are driven into saturation for high transmit power levels.
Next, we note that for given $N_\text{t}$ and $R_\Sigma$, the proposed scheme requires a  significantly lower transmit power than the baseline schemes. 
This is due to the more accurate modelling of the EH circuits which enables the optimization of the instantaneous powers harvested at the user devices.
{Finally, since a higher number of BS antennas leads to beamforming gain and channel hardening, we observe a better system performance in this case.}
	\section{Conclusions}
	We considered two-user MISO WPCNs with non-linear EH circuits at the user devices, where the downlink and uplink transmission phases were utilized for power and information transfer, respectively.
We formulated an optimization problem for the minimization of the transmit power in the downlink with per-user rate constraints in the uplink.
{We provided conditions for the feasibility of the formulated problem and the existence of a trivial solution, respectively.}
Next, for the case when the problem is feasible and the solution is non-trivial, we proved that  three beamforming vectors are sufficient for optimal power transfer in the downlink.
We obtained suboptimal solutions for these vectors via SDR and SCA.
Our simulation results revealed that the proposed WPCN design outperforms baseline schemes based on linear and sigmoidal EH models.
	
	\appendices
	\begin{appendices}
		\renewcommand{\thesection}{\Alph{section}}
		\renewcommand{\thesubsection}{\thesection.\arabic{subsection}}
		\renewcommand{\thesectiondis}[2]{\Alph{section}:}
		\renewcommand{\thesubsectiondis}{\thesection.\arabic{subsection}:}	
		\section{Proof of Proposition 1}
		\label{Appendix:ProofProposition1}
		First, we note that the optimal value $\bar{\tau}^*$ as solution of (\ref{Eqn:OriginalProblem}) can be obtained as $\bar{\tau}^* = \argmin P_\text{DL}^*(\bar{\tau})$, where
\begin{equation}
	P_\text{DL}^*(\bar{\tau}) = \min_{\mathcal{F}} \{P_{\text{DL}} \}
	\label{Eqn:ProofProblemOriginal}
\end{equation}
and $\mathcal{F} = \{\boldsymbol{\tau}, {\boldsymbol{X}}, \boldsymbol{p}^\text{u}, N:\text{(\ref{Eqn:ProblemC1})}, \text{(\ref{Eqn:ProblemC2})}, \text{(\ref{Eqn:ProblemC3})} \}$.
Next, we equivalently reformulate optimization problem in (\ref{Eqn:ProofProblemOriginal}) as follows:
\begin{subequations}
	\begin{align}
		\minimize_{\substack{\boldsymbol{\beta}, \boldsymbol{X}, N}  } \; & \bar{\tau} \sum_{n=1}^{N}  \beta_n \|\boldsymbol{x}_n\|_2^2 \label{Eqn:ProblemIntervalObj} \\
		\subjectto \; & \sum_{n=1}^{N} \beta_n \phi\big(|\boldsymbol{h}_k^H \boldsymbol{x}_n|^2\big) \geq \frac{g_k(\bar{\tau})}{\bar{\tau}}, \forall k  \label{Eqn:ProblemIntervalC1}\\
		&\sum_{n=1}^{N} \beta_n = 1, \label{Eqn:ProblemIntervalC2}
	\end{align}
	\label{Eqn:ProblemInterval}
\end{subequations}
\noindent where $\boldsymbol{\beta} = [\beta_1, \beta_2, \cdots, \beta_N]$ with $\beta_n = \frac{\tau_n}{\bar{\tau}} \geq 0, n\in\{1,2,\cdots,N\},$ and $g_k(\bar{\tau}) = (1-\bar{\tau}) \Big[2^\frac{R^{\text{req}}_k}{1-\bar{\tau}} - 1\Big]{\tilde{\sigma}^2} - \frac{q_k}{T_f}$.
We note that the minimum transmit power $P_\text{DL}^*(\cdot)$ in (\ref{Eqn:ProofProblemOriginal}) can be evaluated as the value of the objective function ({\ref{Eqn:ProblemIntervalObj}}) obtained in the optimal point of problem (\ref{Eqn:ProblemInterval}).
Here, for a given $\bar{\tau}$, function $f_k(\bar{\tau}) \triangleq \frac{g_k(\bar{\tau})}{\bar{\tau}}$ returns the minimum harvested power required at user $k$ for uplink transmission with data rate $R^{\text{req}}_k$.
Since function $\phi(\cdot)$ is bounded, problem (\ref{Eqn:ProblemInterval}) is feasible if and only if $f_k(\bar{\tau}) \leq \varphi(A_s^2), \forall k \in \{1,2\}$.
Hence, problem (\ref{Eqn:OriginalProblem}) is feasible if and only if $\exists \bar{\tau} \in [0,1): f_k(\bar{\tau}) \leq \phi(A_s^2), \, \forall k \in\{1,2\}$.
This concludes the proof.
		\section{Proof of Proposition 2}
		\label{Appendix:ProofProposition2}
		We note that the optimal value $\bar{\tau}^*$ as solution of (\ref{Eqn:OriginalProblem}) can be obtained as $\bar{\tau}^* = \argmin P_\text{DL}^*(\bar{\tau})$, where function $P_\text{DL}^*(\bar{\tau})$ is given by (\ref{Eqn:ProofProblemOriginal}).
Furthermore, optimization problem (\ref{Eqn:ProofProblemOriginal}) can be equivalently reformulated as in (\ref{Eqn:ProblemInterval}).
Let us consider function $f_k(\bar{\tau}) \triangleq \frac{g_k(\bar{\tau})}{\bar{\tau}}, k\in\{1,2\}, $ in (\ref{Eqn:ProblemInterval}).
First, we note that $f_k(\bar{\tau}) \to \infty$ for $\bar{\tau} \to 1$.
Next, if condition (\ref{Eqn:PropositionCondition}) holds for user $k$, we have $g_k(0) < 0$, and hence, $f_k(\bar{\tau}) \to -\infty$ for $\bar{\tau} \to 0$.
Furthermore, if (\ref{Eqn:PropositionCondition}) holds, it can be shown that for the derivative $f'_k(\bar{\tau})$ of function $f_k(\bar{\tau})$, we have $f'_k(\bar{\tau}) > 0, \; \forall \bar{\tau} \in [0,1), \forall k$, and therefore, function $f_k(\bar{\tau})$ is monotonic increasing, $\forall k \in \{1,2\}$.
Let us consider $\bar{\tau}_2 \geq \bar{\tau}_1$, where $\bar{\tau}_n \in [0,1), n\in\{1,2\},$ and denote the feasible sets of (\ref{Eqn:ProblemInterval}) corresponding to $\bar{\tau} = \bar{\tau}_1$ and $\bar{\tau} = \bar{\tau}_2$ by $\mathcal{F}_1$ and $\mathcal{F}_2$, respectively.
We note that if condition (\ref{Eqn:PropositionCondition}) holds $\forall k \in\{1,2\}$, then $\mathcal{F}_2 \subset \mathcal{F}_1$, and hence, $P^*_\text{DL}(\bar{\tau}_2) > P^*_\text{DL}(\bar{\tau}_1), \, \forall \bar{\tau}_1, \bar{\tau}_2 \in [0,1)$.
Thus, we conclude that function $P^*_\text{DL}(\cdot)$ is monotonic increasing, and therefore, the optimal value is $\bar{\tau}^* = 0$.
Hence, the optimal solution of (\ref{Eqn:OriginalProblem}) is trivial, i.e., $N^* = 0, \boldsymbol{x}^* = \boldsymbol{0}, \boldsymbol{\tau}^* = \boldsymbol{0}, {p}^{\text{u}^*}_k = \big[2^{R^{\text{\upshape{req}}}_k} - 1\big]{\tilde{\sigma}^2}, \forall k \in \{1,2\}$.
This concludes the proof.

		\section{Proof of Proposition 3}
		\label{Appendix:ProofProposition3}
		We note that the optimal value $\bar{\tau}^*$ as solution of (\ref{Eqn:OriginalProblem}) can be obtained as $\bar{\tau}^* = \argmin P_\text{DL}^*(\bar{\tau})$, where function $P_\text{DL}^*(\bar{\tau})$ is given by (\ref{Eqn:ProofProblemOriginal}).
Furthermore, optimization problem (\ref{Eqn:ProofProblemOriginal}) can be equivalently reformulated as in (\ref{Eqn:ProblemInterval}).
If problem (\ref{Eqn:OriginalProblem}) is feasible and condition (\ref{Eqn:PropositionCondition}) does not hold for any $k \in\{1,2\}$, in (\ref{Eqn:ProblemInterval}), we have $g_k(0) > 0$ and $f_k(\bar{\tau}) \to \infty$ for $\bar{\tau} \to 0$.
Moreover, it can be shown that for the second derivative of function $f_k(\cdot)$, $f_k''(\bar{\tau}) > 0$ holds $\forall \bar{\tau} \in [0,1)$, and hence, function $f_k(\bar{\tau} )$ is convex and reaches its minimum at $\bar{\tau} = \bar{\tau}^\text{max}_k$.
Furthermore, it can be shown that $f_k(\bar{\tau}_k^\text{max}) > 0$, $\forall k$.
Thus, for $\bar{\tau} \in [\bar{\tau}^\text{max}, 1]$, function $P_\text{DL}^*$ is monotonic increasing and $\bar{\tau}^* \in [0, \bar{\tau}^\text{max}]$.
Furthermore, if we assume that $\bar{\tau} \in [0, \bar{\tau}^\text{min})$, then $\exists k \in \{1,2\}$, such that $f_k(\bar{\tau}) > \varphi(A_s^2)$ and problem (\ref{Eqn:ProblemInterval}) is infeasible.
Hence, the optimal $\bar{\tau}^* \in [\bar{\tau}^\text{min}, \bar{\tau}^\text{max}]$.

Finally, we prove that for the solution of (\ref{Eqn:ProblemInterval}), the optimal number of transmit vectors satisfies $N^* \leq 3$.
To this end, let us consider a function $\tilde{P}_\text{DL}(\bar{\tau}, \boldsymbol{X}, N)$ defined as
\begin{equation}
	\tilde{P}_\text{DL}(\bar{\tau}, \boldsymbol{X}, N) = \min_{\boldsymbol{\beta}} \bar{\tau} \sum_{n=1}^{N}  \beta_n \|\boldsymbol{x}_n\|_2^2 \quad \subjectto \; \text{(\ref{Eqn:ProblemIntervalC1})}, \text{(\ref{Eqn:ProblemIntervalC2})}.
	\label{Eqn:ProblemIntervalRef}
\end{equation}
For any given $\bar{\tau} \in [\bar{\tau}^\text{min}, \bar{\tau}^\text{max}]$, the optimal transmit power in (\ref{Eqn:ProofProblemOriginal}) can be obtained as ${P}^*_\text{DL}(\bar{\tau}) = \min_{ \{ \boldsymbol{X}, N \} } \tilde{P}_\text{DL}(\bar{\tau}, \boldsymbol{X}, N)$.
Optimization problem (\ref{Eqn:ProblemIntervalRef}) is linear in $\boldsymbol{\beta}$ and involves $N$ variables and $3$ constraints.
We note that the solution of a linear optimization problem with $\tilde{N}$ variables and $\tilde{K} \leq \tilde{N}$ constraints is a vertex of the polytope defined by $\tilde{K}$ constraints, and thus, has at most $\tilde{K}$ non-zero elements \cite{Nocedal2006}.
Thus, for the optimal solution of (\ref{Eqn:ProblemIntervalRef}) and, hence, (\ref{Eqn:ProblemInterval}) and (\ref{Eqn:OriginalProblem}), at most $N^* = 3$ time slots have non-zero lengths. 
This concludes the proof.

		\section{Proof of Proposition 4}
		\label{Appendix:ProofProposition4}
		First, we rewrite convex problem (\ref{Eqn:AlgorithmProblem}) equivalently as follows:
\begin{subequations}
	\begin{align}
		\minimize_{\substack{\boldsymbol{V}_{\hspace*{-2pt}1}, \boldsymbol{V}_{\hspace*{-2pt}2}, \boldsymbol{V}_{\hspace*{-2pt}3}, \boldsymbol{\beta}}}  \quad& \sum_{n=1}^{3} \Tr{\boldsymbol{V}_{\hspace*{-2pt}n}}  \\
		\subjectto \quad &\bar{\bar{\tau}} \bar{\tau} \sum_{n=1}^{3} \big[ \phi'(\boldsymbol{h}^H_k \boldsymbol{W}^{(t)}_{\hspace{-2pt}n} \boldsymbol{h}_k) \Tr{\boldsymbol{H}_k \boldsymbol{V}_{\hspace*{-2pt}n}} + \beta_n \hat{\psi}(\boldsymbol{W}^{(t)}_{\hspace{-2pt}n})\big] + \bar{q}_k/T_\text{f} \geq \rho_k^{\text{req}}, \; \forall k, \label{Eqn:RankProofProblemC1} \\
		&\Tr{\boldsymbol{H}_k \boldsymbol{V}_{\hspace*{-2pt}n}} - \beta_n A_s^2 \leq 0, \; \forall n, k, \label{Eqn:RankProofProblemC2} \\
		&\sum_{n=1}^{3} \beta_n = 1 \\
		& \boldsymbol{V}_{\hspace*{-2pt}n} \in \mathcal{S}_{+}, \; \forall n. \label{Eqn:RankProofProblemC4}
	\end{align}
	\label{Eqn:RankProofProblem}
\end{subequations}
\noindent\hspace*{-4pt}Since (\ref{Eqn:RankProofProblem}) is feasible and convex, strong duality holds and the gap between (\ref{Eqn:RankProofProblem}) and its dual problem is equal to zero \cite{Boyd2004}.
We express the Lagrangian of (\ref{Eqn:RankProofProblem}) with respect to $\boldsymbol{V}_n$ as follows:
\begin{align}
	\mathcal{L}(\boldsymbol{V}_{\!\!n}) =  \Tr{\boldsymbol{V}_{\!\!n}} - \sum_k (\mu_k \bar{\psi}_k - \lambda_k) \Tr{\boldsymbol{H}_k \boldsymbol{V}_{\!\!n}}  - \Tr{\boldsymbol{Y}_{\!\!n}\boldsymbol{V}_{\!\!n}} + {\gamma},
	\label{Eqn:RankProofLagrangian}
\end{align}
\noindent where ${\mu_k}$ and ${\lambda_k}$, $k \in \{1,2\}$, are Lagrangian multipliers associated with constraints (\ref{Eqn:RankProofProblemC1}) and (\ref{Eqn:RankProofProblemC2}), respectively, and $\bar{\psi}_k$ collects all terms that do not depend on $\boldsymbol{V}_{\!\!n}$.
{Here, $\boldsymbol{Y}_{\!\!n}$ is the Lagrangian multiplier associated with the constraint on $\boldsymbol{V}_{\!\!n}$ in (\ref{Eqn:RankProofProblemC4})} and $\gamma$ accounts for all terms that do not involve $\boldsymbol{V}_{\!\!n}$.
We note that the Karush-Kuhn-Tucker (KKT) conditions are satisfied for the optimal solution of (\ref{Eqn:RankProofProblem}) with respect to $\boldsymbol{V}_{\!\!n}$, denoted by $\boldsymbol{V}_{\!\!n}^*$, and the solutions ${\mu}_k^*, {\lambda}_k^*, k\in\{1,2\}$, and $\boldsymbol{Y}_{\!\!n}^*$ of the corresponding dual problem.
The KKT conditions are given by
\begin{subequations}
	\begin{align}
		&\triangledown \mathcal{L}(\boldsymbol{V}_{\!\!n}^*) = \boldsymbol{0}_{N_\text{t} \times N_\text{t}}
		\label{Eqn:RankProofKKTGradient}\\
		& \mu_k^* \geq 0, \lambda_k^* \geq 0, \boldsymbol{Y}_{\!\!n}^* \succeq 0, \; \forall k \in \{1,2\}
		\label{Eqn:RankProofKKTPositive}\\
		& \boldsymbol{Y}_{\!\!n}^* \boldsymbol{V}_{\!\!n}^* = \boldsymbol{0}_{N_\text{t} \times N_\text{t}},
		\label{Eqn:RankProofKKTSemidef}
	\end{align}
	\label{Eqn:RankProofKKT}
\end{subequations}
\noindent\hspace*{-4pt}where $\triangledown \mathcal{L}(\boldsymbol{V}_{\!\!n}^*)$ and $\boldsymbol{0}_{N_\text{t} \times N_\text{t}}$ denote the gradient of $\mathcal{L}(\boldsymbol{V}_{\!\!n}^*)$ evaluated in $\boldsymbol{V}_{\!\!n}^*$ and the square all-zero matrix of size $N_\text{t}$, respectively.
Next, we express condition (\ref{Eqn:RankProofKKTGradient}) as follows
\begin{equation}
	\boldsymbol{Y}_{\!\!n}^* = \boldsymbol{I}_{N_\text{t}} - \boldsymbol{\Delta}, \label{Eqn:RankProofMainEquality} 
\end{equation}
\noindent where $\boldsymbol{\Delta} = \sum_k (\mu_k^* \bar{\psi}_k - \lambda_k^*) \boldsymbol{H}_k^\top$.
Let us now investigate the structure of $\boldsymbol{\Delta}$.
We denote the maximum eigenvalue of $\boldsymbol{\Delta}$ by $\delta^\text{max} \in \mathbb{R}$.
Due to the randomness of the channel, with probability $1$, only one eigenvalue of $\boldsymbol{\Delta}$ has value $\delta^\text{max}$.
Observing (\ref{Eqn:RankProofMainEquality}), we note that if $\delta^\text{max}<1$, then $\boldsymbol{Y}_{\!\!n}^*$ is a positive semidefinite matrix with full rank.
In this case, (\ref{Eqn:RankProofKKTSemidef}) yields $\boldsymbol{V}_{\!\!n}^* = \boldsymbol{0}_{N_\text{t} \times N_\text{t}}$ and, hence, $\rank \boldsymbol{V}_{\!\!n}^* = 0$.
Furthermore, if $\delta^\text{max}>1$, then $\boldsymbol{Y}_{\!\!n}^*$ is not a positive semidefinite matrix, which contradicts (\ref{Eqn:RankProofKKTPositive}).
Finally, if $\delta^\text{max} = 1$, then $\boldsymbol{Y}_{\!\!n}^*$ is a positive semidefinite matrix with $\rank \{\boldsymbol{Y}_{\!\!n}^*\} = N_\text{t} - 1$.
Then, applying Sylvester's rank inequality to (\ref{Eqn:RankProofKKTSemidef}), we have
\begin{align}
	0 = \rank \{\boldsymbol{Y}_{\!\!n}^* \boldsymbol{V}_{\!\!n}^*\} \geq \rank \{\boldsymbol{Y}_{\!\!n}^*\} &+ \rank \{\boldsymbol{V}_{\!\!n}^*\} - N_\text{t} \nonumber \\ &= \rank \{\boldsymbol{V}_{\!\!n}^*\} - 1.
\end{align}
Thus, we conclude that $\rank \{\boldsymbol{V}_{\!\!n}^*\} \leq 1$, $\forall n \in\{1,2,3\}$. This concludes the proof.

	\end{appendices}

	\bibliographystyle{IEEEtran}
	\bibliography{Final}

\end{document}